\documentclass[subeqn,final]{siamltex}
\usepackage{amsfonts}
\usepackage{setspace}

\newcommand{\norm}[2]{\|#1\|_#2} 
\newcommand{\stconv}{\star_{s,t}} 
\newcommand{\bll}{ {\bm \lambda} }

\newcommand\eqref[1]{(\ref{#1})}
\newcommand{\mani}{L^2(\Omega \times \left[0,T\right])}
\newcommand{\sball}{{\mathbb B}_{(c\alpha,c\beta)}(|{\bf r}|)}
\newcommand{\tball}{{\mathbb H}_{(\alpha,\beta)}({t})}
\newcommand{\sballi}{{\mathbb B}_{(c\alpha,c\beta)}(|{\bf r}_0-{\bf r}'|)}
\newcommand{\tballi}{{\mathbb H}_{(\alpha,\beta)}({t_0-t'})}
\newcommand{\tmf}[1]{\tilde{#1}(\omega)}
\newcommand{\spf}[1]{\tilde{#1}({\bm \lambda})}
\usepackage{tikz,bm,subfig}
\usetikzlibrary{calc,fadings,decorations.pathreplacing}

\title{On late-time Stability of Time Domain Integral Equations for Electromagnetics \thanks{This work was supported by the National Science Foundation through grants CCF:1018576 and CMMI:1250261 . A. J. Pray was also supported by the DoD SMART program under grant\# N00244-09-1-0081}}

\author{N.~V. Nair \thanks{Department of Electrical and Computer Engineering, Michigan State University, East Lansing, MI 48824 ({\tt nairn@msu.edu}).} \and A.~J. Pray \thanks{Department of Electrical and Computer Engineering, Michigan State University, East Lansing, MI 48824 ({\tt prayandr@msu.edu}).}
\and B. Shanker\thanks{Department of Electrical and Computer Engineering, Michigan State University, East Lansing, MI 48824 ({\tt bshanker@egr.msu.edu}).}}

\begin{document}

\maketitle

\begin{abstract}
The problem of late time instability in time domain integral equations for electromagnetics is longstanding. While several techniques have been suggested for addressing this problem, they either require impractically high degrees of freedom in the basis function or an analytical computation of matrix elements. The authors recently proposed a method that demonstrates stability without requiring either of these. The paper, however, does not present theoretical foundations for the choice of, or a rigorous error bounds on the approximation of,  the bilinear form. This paper complements the authors' previous work by presenting a construction of the bilinear form based on the minimization of the energy in the system and a proof for the bounds on the approximation. We present results on the bounds developed and few sample scattering results that demonstrate the stability of the proposed scheme. \end{abstract}



\pagestyle{myheadings}
\thispagestyle{plain}
\markboth{NAIR, PRAY AND SHANKER}{STABILITY OF TIME DOMAIN INTEGRAL EQUATIONS}

\section{Introduction and examples}
While time domain integral equation (TDIE) methods have been used to solve electromagnetic scattering problems for a few decades now,  instability has been a significant debilitating factor, limiting its practical applications. Several methods have been proposed to address this issue mathematically, from initial work in acoustics \cite{Ha-Duong1987,Ha-Duong2003,Ha-Duong2003a}, to more recent work in 2-D \cite{Pujols1991} and 3-D \cite{Terrasse1993,Bachelot1995,Bachelot1995a} electromagnetics.  While these methods provide a mechanism to construct solutions that are stable for long simulation times \cite{Ha-Duong2003a}, they require high degrees of smoothness in the basis function for the proofs to be valid \cite{Bachelot1995}.  

Concurrent with the developments in the mathematical framework, there has also been several numerical attempts to solve the problem, viz. obtain stable solutions to transient electromagnetic scattering using TDIEs \cite{Ha-Duong1987,Ha-Duong2003,Ha-Duong2003a,Shanker2009,Shi2011}. While these techniques have been quite successful in generating stable solutions for a variety of challenging geometries, they rely on analytical computation of the matrix elements. However, this precludes the use of either acceleration methodologies or higher order geometry representations, both of which are essential for application to realistic geometries. 

Recently, the authors developed a method \cite{Pray2012} that relies on a separable expansion of the convolution with the retarded potential Green's function. This scheme approximates the convolution of the Green's function with the temporal basis function in any given closed domain using a summation of smooth polynomial functions and has been numerically shown to generate stable solutions for a large class of geometries that could previously {\em only} be stabilized via the exact integration techniques. While this method does not suffer from either of the disadvantages of the schemes developed in \cite{Shanker2009,Shi2011} and provided ample empirical evidence of utility, \cite{Pray2012} did not provide a rigorous mechanism for truncating the separable expansion. The present work answers this need.

This paper aims to provide a theoretical footing to the scheme developed in \cite{Pray2012,Pray2012a}. To this end we will (1) construct a bilinear form for the electric field integral equation for time domain electromagnetics whose solution leads to a minimization of the energy in the system and (2) provide a rigorous mechanism to truncate the polynomial expansion developed in \cite{Pray2012,Pray2012a}. While we will provide some numerical results that demonstrate the validity of our approximations and the stability of the scheme developed, we will defer to \cite{Pray2012} for more detailed results on stability.

The rest of this paper will proceed as follows. In Section \ref{sec:prob}, we formally state the problem following which, Section \ref{sec:stabtd} will construct an energy-based bilinear form of the TDEFIE, inspired by \cite{Ha-Duong1987,Ha-Duong2003,Ha-Duong2003a,Bachelot1995,Bachelot1995a}. The next section will describe the new solution scheme developed in \cite{Pray2012} and provide the theoretical machinery for choosing bounds on the expansion. Finally, Section \ref{sec:results} will provide numerical examples that validate this technique.   

\section{Variational Formulation based on the energy \label{sec:prob}}

Consider a perfect electrically conducting (PEC) object that resides in free space. Let $\Omega$ denote the surface of this object that occupies a domain $D_-$, and let $D_+ = \mathbb{R}^3/D_-$. Assume that every point on the surface is equipped with a unique normal ${\hat n}_\pm ({\bf r})$ pointing into $D_\pm$ such that \begin{equation}
  {\hat n} = {\hat n}_+ = - {\hat n}_-.
  \label{eq:nrm}
\end{equation}
\begin{figure}[!h]
     \centering
\begin{tikzpicture}[scale=0.5,thick]
    \begin{scope}[scale=0.75,thick]
    \pgftransformrotate{-40}
        \draw(0,0) ellipse (2cm and 4cm);
    \draw[->](1.6,2.5) .. controls (2.0,3.0) and (2.4,3.0) .. (3.0,2.5);
    \draw[->](0.0,4.0)--(0.0,5.0);
    \draw (0.0,5.8) node {$\hat{\bf n}$};
    \draw(3.5,2.4) node {$\Omega$};
    \draw(0.2,2.5) node {$D_-$};
   \draw(4.2,0.25) node {$D_+$};
    \draw[->](-6,0)--(-4,0);
    \draw(-6,1) node {\footnotesize $\{{\bf E}^i({\bf r},t)$, ${\bf H}^i({\bf r},t)\}$};
    \draw[->](4.5,4.5)--(5.0,6.65);
    \draw(5.5,8.0) node {\footnotesize$\{{\bf E}^s({\bf r},t)$, ${\bf H}^s({\bf r},t)\}$};
    \end{scope}
\end{tikzpicture}
\end{figure}

An electric field parameterized by $\left \{ {\bf E}^i ({\bf r},t), {\bf H}^i ({\bf r},t) \right \}$ and bandlimited to $\omega_{max}$ is incident on this body. It is assumed that the fields are quiescent for $t < 0$. The incident field induces unknown currents ${\bf J} ({\bf r},t) \forall {\bf r} \in \Omega$ that produce scattered fields $\left \{ {\bf E}^s ({\bf r},t), {\bf H}^s ({\bf r},t) \right \}$.  We denote the total electric and magnetic fields in $D_\pm$ as ${\bf E}_\pm({\bf r},t)$ and ${\bf H}_\pm({\bf r},t)$. The problem then is to solve for the current ${\bf J}({\bf r},t)$ on the surface of the scatterer ($\Omega$). To this end, we can setup an equivalent problem by setting the fields inside the scatterer ($D_-$) to \begin{subequations}
  \label{eq:eqp}
  \begin{eqnarray}
    {\bf E}_-({\bf r},t) &=& - {\bf E}^i ({\bf r},t) ~\forall~ {\bf r} \in D_-,\\
    {\bf H}_-({\bf r},t) &=& - {\bf H}^i ({\bf r},t) ~\forall~ {\bf r} \in D_-,
  \end{eqnarray}
  and the fields outside the scatterer ($D_+$) to  
  \begin{eqnarray}
    {\bf E}_+({\bf r},t) &=& {\bf E}^s ({\bf r},t) ~\forall~ {\bf r} \in D_-,\\
    {\bf H}_+({\bf r},t) &=& {\bf H}^s ({\bf r},t) ~\forall~ {\bf r} \in D_-,
  \end{eqnarray}
  Then the current ${\bf J}({\bf r},t)$  on the scatterer surface can be written as \begin{equation}
    {\bf J}({\bf r},t) \doteq \hat{\bf n} \times \big[{\bf H}^s({\bf r},t) + {\bf H}^i({\bf r},t)\big].
  \label{eq:cur}
\end{equation}
\end{subequations}
Next, inspired by \cite{Ha-Duong2003,Ha-Duong2003a,Terrasse1993}, we derive a variational formulation relating the sum of the incident field energy interior and the scattered field energy exterior to $D$.    
\section{An energy inspired variational formulation \label{sec:stabtd}}
The electromagnetic energy $E(t)$ at an instant $t$, $\mathbb{R}^3$ can be written as $E(t) = E_+(t) + E_-(t)$, where $E_\pm(t)$  is  defined as 
  \begin{equation}
    E_\pm (t) = \frac{1}{2} \left [ \varepsilon_0 \int_{D_\pm} d {\bf r} \left |{\bf E}_\pm({\bf r},t) \right |^2 + \mu_0 \int_{D_\pm} d{\bf r}' \left | {\bf H}_\pm({\bf r},t) \right |^2 \right ].
    \label{eq:insten2}
  \end{equation} From Poynting's theorem \cite{Stratton1941}, we have the rate of change of total energy as  \begin{equation}
  \frac{\partial}{\partial t} E_\pm (t) = \int_{\Omega} d {\bf r} {\bf E}_\pm ({\bf r},t)\cdot \left (  {\hat n}_\pm \times {\bf H}_\pm({\bf r},t)\right) 
   \label{eq:poynt}
 \end{equation}
 Combining \eqref{eq:poynt} with \eqref{eq:nrm} and \eqref{eq:eqp}, after some manipulation, yields  
 \begin{subequations}
   \begin{equation}
     \frac{\partial}{\partial t} E (t) = \frac{\partial}{\partial t}\left[E_+(t) + E_-(t)\right] = \int_{\Omega} d{\bf r} {\bf E}^s({\bf r},t) \cdot {\bf J}({\bf r},t).
   \label{eq:ej}
 \end{equation}
 Thus, the total energy in the system at some time $T < \infty$ is given by\begin{equation}
   E(T) = \int_{-\infty}^T  dt \frac{\partial}{\partial t} E (t) =\int_{0}^T  dt \int_{\Omega} d{\bf r} {\bf E}^s({\bf r},t) \cdot {\bf J}({\bf r},t),
   \label{eq:concl}
 \end{equation}
	since ${\bf E}^s({\bf r},t) = 0 \forall t < 0$.
 \end{subequations}

 The scattered fields ${\bf E}^s({\bf r},t)$ and ${\bf H}^s({\bf r},t)$ can be related to the surface current ${\bf J}({\bf r},t)$ via integral operators ${\mathcal L}$ and ${\mathcal K}$ as 
\begin{subequations}
  \label{eq:IE}%
\begin{eqnarray}
  \label{eq:EFIE}
  {\bf E}^s ({\bf r},t) &=& {\mathcal L} \circ \{{\bf J}({\bf r},t)\} \\
  &\doteq& -\frac{\mu_0}{4 \pi} \partial_t {\mathcal G} \circ \{{\bf J}({\bf r},t)\} - \frac{1}{4 \pi \varepsilon_0} \nabla \int_{-\infty}^{t} d\tau {\mathcal G} \circ \{\nabla \cdot {\bf J}({\bf r},\tau)\}\nonumber 
\end{eqnarray}
and
\begin{equation}
  \label{eq:MFIE}
  \hat{\bf n} \times {\bf H}^s ({\bf r},t) = {\mathcal K} \circ \{ {\bf J}({\bf r},t)\}  \doteq \hat{\bf n} \times \nabla \times \frac{1}{4\pi}{\mathcal G} \circ \{{\bf J}({\bf r},t)\}   
\end{equation}
where ${\mathcal G}$ is the retarded potential operator defined for functions ${\bf J}({\bf r},t)$ as  
\begin{equation}
  {\mathcal G} \circ \{{\bf J}({\bf r},t)\} \doteq \int_\Omega d{\bf r}'  \int_{-\infty}^{\infty} dt' \frac{\delta(t-t' - \frac{|{\bf r}-{\bf r}'|}{c})}{|{\bf r}-{\bf r}'|} {\bf J}({\bf r}',t')
\end{equation}
\end{subequations}


Following the development in \cite{H-Duong2003a}, we note that equation \eqref{eq:concl} suggests the construction of a space time variational form as \begin{equation}
  \int_{0}^T dt\int_\Omega d{\bf r}~{\bf j}({\bf r},t)\cdot \left [{\bf E}^s({\bf r},t) \right ] = \int_{0}^T dt \int_\Omega d{\bf r}~{\bf j}({\bf r},t)\cdot \left[-{\bf E}^i({\bf r},t)\right]
  \label{eq:varfinal}
\end{equation}
for a test function ${\bf j}({\bf r},t) \in \mani$. We can write this variational form as \begin{equation}
  {\mathcal B}_{e}\left[{\bf j},{\bf J}\right] =  \int_{0}^T dt \int_\Omega d{\bf r}~{\bf j}({\bf r},t)\cdot \Big[-{\bf E}^i({\bf r},t) \Big], 
  \label{ebil}
\end{equation} where the bilinear form ${\mathcal B}_e\left[{\bf j},{\bf J}\right]$ can be defined using the operator notation introduced in \eqref{eq:IE},  \begin{equation}
  {\mathcal B}_e\left[{\bf j}, {\bf J}\right] = \int_{0}^T dt \int_\Omega d{\bf r}~{\bf j}({\bf r},t)\cdot \Big[{\mathcal L} \circ \{{\bf J}({\bf r},t)\}\Big];
  \label{eq:vare}
\end{equation} the upper limits on the temporal integral can be restricted to $T$ due to the finite velocity of propagation. We can construct a similar bilinear form for the magnetic field equation \eqref{eq:MFIE} which can be written as \begin{equation}
  {\mathcal B}_h\left[{\bf j},{\bf J}\right] =  \int_{0}^T dt \int_\Omega d{\bf r}~{\bf j}({\bf r},t)\cdot \Big[ \hat{\bf n} \times {\bf H}^i({\bf r},t)\Big]=\int_{0}^T dt \int_\Omega d{\bf r}~{\bf j}({\bf r},t)\cdot \Big[\left({\mathcal I} - {\mathcal K}\right)  \circ \{{\bf J}({\bf r},t)\}\Big]
  \label{hbil}
\end{equation}
For some parameter $\alpha > 0$, we can now construct the variational form for the combined field integral equation as \begin{equation}
  {\mathcal B}\left[{\bf j},{\bf J}\right] =  \int_{0}^T dt \int_\Omega d{\bf r}~{\bf j}({\bf r},t)\cdot \Big[-{\bf E}^i({\bf r},t) + \alpha \hat{\bf n} \times {\bf H}^i({\bf r},t)\Big],
  \label{eq:bil}
\end{equation} where the bilinear form ${\mathcal B}\left[{\bf j},{\bf J}\right]$ can be defined as \begin{equation}
  {\mathcal B}\left[{\bf j}, {\bf J}\right] = \int_{0}^T dt \int_\Omega d{\bf r}~{\bf j}({\bf r},t)\cdot \Big[{\mathcal L} \circ \{{\bf J}({\bf r},t)\} + \alpha\left({\mathcal I} - {\mathcal K}\right)  \circ \{{\bf J}({\bf r},t)\}\Big].
  \label{eq:var}
\end{equation} Finally, we define the notation \begin{equation}
 	<{\bf u}, {\bf v}> = \int_{0}^T dt \int_\Omega d{\bf r}~{\bf u}({\bf r},t)\cdot {\bf v}({\bf r},t) 
	\end{equation}
for functions ${\bf u}, {\bf v} \in \mani$.
We can now extend ideas obtained in \cite{Ha-Duong2003,Ha-Duong2003a,Terrasse1993} to show coercivity of the bilinear form and stability of the solution scheme. To this end we have the following relationship between the bilinear form in \eqref{eq:bil} and the energy.  
\begin{equation}
  {\mathcal B}[{\bf J},{\bf J}] = E(T) + <{\bf J},\alpha \left({\mathcal I}-{\mathcal K}\right)\circ \{{\bf J}\}>.
  \label{eq:bilen}
\end{equation}Combining this with the compactness of the MFIE (${\mathcal K}$) operator \cite{Hsiao1997}, we have the following G{a}rding inequality for the coercivity of the bilinear form ${\mathcal B}$. 
\begin{proposition}
  Given the bilinear form, ${\mathcal B}[{\bf j},{\bf J}]$, $\exists$ a constant $C > 0$ such that
\begin{equation}
  {\mathcal B}[{\bf J},{\bf J}] + \alpha <{\bf J}, \left({\mathcal K}) \circ \{{\bf J}\}\right)> ~\ge  C~\Vert{\bf J}\Vert^2 
  \label{eq:coerc}
\end{equation}
for norms in $\mani$.
\end{proposition}

The proof of the proposition follows naturally from the bounded nature of the ${\mathcal K}$ operator and the positivity of the energy.

\section{Numerical evaluation of the bilinear form \label{sec:eval}}
    Accurate evaluation of the bilinear form in \eqref{eq:bil} is important to ensure stability \cite{Ha-Duong2003}. In practice, a discretization subspace $V_{h,\tau}$ is chosen as that spanned by a set of space-time basis functions of the form ${\bf S}({\bf r})T(t)$, where ${\bf S}({\bf r})$ are spatial basis functions (usually the Rao-Wilton-Glisson or Thomas-Raviart \cite{Nedelec1980, Rao1982} spaces), defined on a triangular tessellation of $\Omega$, say $\{\Omega_s\}$.  $T(t)$ are usually shifted piecewise Lagrange polynomials with support in $\Omega_t \subset (0,T)$. $V_{h,\tau}$ is indexed on the average size of the spatial and temporal functions, i.e. $h \approx diam(\Omega_s)$ and $\tau \approx diam(\Omega_t)$. While the demonstration of stability developed in \cite{Ha-Duong2003} could be extended to electromagnetics \cite{Bachelot1995}, that is not the focus of this paper. Instead, here we focus on the construction of a rigorous mechanism to evaluate the bilinear form above.  
    
    The piecewise nature of the Lagrange polynomials implies that accurate computation of the bilinear form requires identification of the domains of continuity of the integrand \cite{Pray2012,Pray2012a}. It has been shown that the most accurate way to evaluate the integral is to compute the correlation between the source and observation domain  \cite{Shi2011}, identify regions of continuity and evaluate the potential integral on each of these regions. However, it has also been shown in \cite{Ha-Duong1987,Ha-Duong2003,Shanker2009} that evaluating the integral accurately over the source domain and collocation over the observation domain is sufficient to lead to stable solutions. This technique is tantamount to to finding arcs of intersection between arbitrary triangles (on which the source resides) and spheres (centered at the point of observation). While tedious, this can be done for piecewise flat triangles \cite{Shanker2009,Shi2011}.  

In \cite{Pray2012,Pray2012a} the authors developed a technique to compute the bilinear form numerically, without the need to identify domains of piecewise continuity of the integrand. This was achieved by approximating the space time convolution of the Green's kernel with the basis function using a separable expansion in space and time. In this section, we will review the expansion and provide a mechanism to truncate the expansion for given error criteria. 
    
We assume a set of functions of the form ${\bf j}({\bf r},t) \doteq{\bf S}({\bf r})T(t)$ defined on a simplicial tessellation of the domain such that ${\bf S}({\bf r}) = 0 ~\forall {\bf r} \notin \Omega_s$ and $T(t)= 0 ~\forall t \notin \Omega_t$. The bilinear form then involves convolutions of the form.  
  \begin{subequations}
    \begin{equation}
      \frac{\delta\left(t-\frac{|{\bf r}|}{c}\right)}{|{\bf r}|}\star_{s,t}{\bf S}({\bf r})T(t) = \int_{\Omega_s} d{\bf r}'{\bf S}({\bf r}') \int_{\Omega_t}dt' T(t') \frac{\delta\left(t-t'-\frac{|{\bf r}-{\bf r}'|}{c}\right)}{|{\bf r}-{\bf r}'|},
    \label{eq:maineq}
  \end{equation}
  which can be rewritten as  
\begin{equation}
  {\mathcal I}_0(t,{\bf r}) =  \int_{\mathbb{R}^3} d{\bf r}'\Pi_{\Omega_s} {\bf S}({\bf r}') \int_{-\infty}^{\infty} dt' \Pi_{\Omega_t} T(t') \frac{\delta\left(t-t'-\frac{|{\bf r}-{\bf r}'|}{c}\right)}{|{\bf r}-{\bf r}'|},
  \end{equation}
  where $\Pi_{\Omega_s}$ and $\Pi_{\Omega_t}$ denote indicator functions on $\Omega_s$ and $\Omega_t$ respectively. Let  $\sball$ be the spherical shell in space of radii $c\alpha$ and $c\beta$, and $\tball$ be a pulse function in time from $t= \alpha$ to $t=\beta$, such that $\Pi_{\Omega_s} \subset \sball$ and $\Pi_{\Omega_t} \subset \tball$. For a given observation point ${\bf r}_0$ and time $t_0$, we have,   \begin{eqnarray}
    &{\mathcal I}_0(t_0,{\bf r}_0) = \\ &\int_{\mathbb{R}^3} d{\bf r}'\Pi_{\Omega_s} S({\bf r}')\sballi \int_{-\infty}^{\infty} dt' \Pi_{\Omega_t} T(t')\tballi \frac{\delta\left(t_0-t'-\frac{|{\bf r}_0-{\bf r}'|}{c}\right)}{|{\bf r}_0-{\bf r}'|}\nonumber ,
    \label{eq:firstapp}
  \end{eqnarray}
  
\end{subequations}  
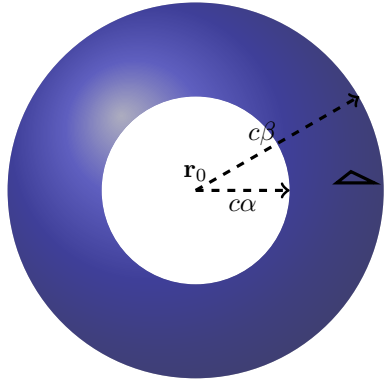
\begin{figure}
  \centering
\begin{tikzpicture}
  \begin{scope}[scale=0.5]
  \def\R{2.5};
\fill[ball color=blue] (0,0) circle (\R); 
\fill[opacity=0.5,ball color=blue] (0,0) circle (5); 
\fill[color=white] (0,0) circle (2.5); 
\draw[->,dashed,very thick] (0,0) -- (2.5,0);
\draw[->,dashed,very thick] (0,0) -- (4.35,2.5);
\draw(1.25,0)node[below]{$c\alpha$};
\draw(1.77,0.95)node[above]{$c\beta$};
\draw[very thick](3.75,0.2)--(4.13,0.5)--(4.75,0.2)--cycle;
\node(0,0)[above]{${\bf r}_0$};
\end{scope}
\end{tikzpicture}
\caption{Domain of the spatial integrals}
\end{figure}

Equation \eqref{eq:firstapp} can be re-written using the Stone-Wierstrass's theorem as \begin{eqnarray}
  &{\mathcal I}_0(t_0,{\bf r}_0) = \\ &\int_{\mathbb{R}^3} d{\bf r}'\Pi_{\Omega_s} S({\bf r}')\sballi \int_{-\infty}^{\infty} dt' \Pi_{\Omega_t} T(t')\tballi \times \nonumber \\ &\sum_{l=0}^{\infty}\frac{(2l+1)}{2}k_1\frac{P_n\left(k_1\frac{|{\bf r}_0-{\bf r}'|}{c} + k_2\right)}{|{\bf r}_0-{\bf r}'|}P_n\left(k_1(t_0 -t')+k_2\right) \nonumber ,
  \label{eq:expn}
\end{eqnarray}

Note, the polynomials $P_n(k_1(t_0-t)+k_2)$ and $P_n(k_1\frac{|{\bf r}_0-{\bf r}'|}{c} + k_2)$ in \eqref{eq:expn} are continuous over $\tballi$ and $\sballi$ respectively, and specifically, over the domains of integration $\Omega_s$ and $\Omega_t$. This permits numerical evaluation of the integral without having to resort to identifying domains of continuity. While we have demonstrated the technique for the source integral, we note that the same scheme can be extended to both source and testing integrals via computation of the geometric correlation function defined in \cite{Shi2011}.
\subsection{Truncation of the expansion}

The expansion in \eqref{eq:expn} is exact. However, in practice, this procedure is useful only if the summation can be truncated to some $N$.  Thus the aim is to evaluate the error \begin{eqnarray}
  \varepsilon_N({\bf r},t) &\doteq& \\ & & \left|\Pi_{\Omega_s}({\bf r})S({\bf r}) \Pi_{\Omega_t}(t)T(t) \stconv \sball\tball \frac{\delta\left(t_0-t'-\frac{|{\bf r}_0-{\bf r}'|}{c}\right)}{|{\bf r}_0-{\bf r}'|} \right. \nonumber \\ 
   &-& \left.\Pi_{\Omega_s}({\bf r})S({\bf r}) \Pi_{\Omega_t}(t)T(t) \stconv \sball\tball \sum_{l=0}^{N}\frac{(2l+1)}{2}k_1\frac{P_n\left(k_1\frac{|{\bf r}_0-{\bf r}'|}{c} + k_2\right)}{|{\bf r}_0-{\bf r}'|}P_n\left(k_1(t_0 -t')+k_2\right) \right| \nonumber. 
  \label{eq:err}
\end{eqnarray}
In order to derive a rigorous bound on the error \eqref{eq:err}, we first define the following 
\begin{definition}
  Let the temporal Fourier transform of  $T(t)$ be \begin{equation}
    \tmf{T} = {\mathcal F}\{T(t)\} \doteq \int_{-\infty}^{\infty} dt T(t) e^{j\omega t},
    \label{eq:tf}
  \end{equation}
  and the spatial Fourier transform of a function $S({\bf r})$ be \begin{equation}
    \spf{S} = {\mathcal F}\{S({\bf r})\} \doteq \int_{\mathbb{R}^{3}} d{\bf r} S({\bf r}) e^{j {{\bm \lambda}} \cdot {\bf r}}. 
    \label{eq:sf}
  \end{equation}
  Using these definitions it is possible to re-write equation \eqref{eq:err} as \begin{eqnarray}
    &&\tilde{\varepsilon}_N(\omega,{\bm \lambda})\doteq \\ && \left|{\mathcal F}\{\Pi_{\Omega_s}({\bf r})S({\bf r}) \Pi_{\Omega_t}(t)T(t)\}  {\mathcal F}\left\{\sball\tball \frac{\delta\left(t_0-t'-\frac{|{\bf r}_0-{\bf r}'|}{c}\right)}{|{\bf r}_0-{\bf r}'|} \right\} \right. \nonumber \\ 
    &-& \left. {\mathcal F}\left\{\sum_{l=0}^{N}\frac{(2l+1)}{2}k_1\frac{P_l\left(k_1\frac{|{\bf r}_0-{\bf r}'|}{c} + k_2\right)}{|{\bf r}_0-{\bf r}'|}P_l\left(k_1(t_0 -t')+k_2\right) \right\} \right| \nonumber.
    \label{eq:farr}
  \end{eqnarray}
 Next, we define spatial and temporal bandwidths of interest $ {\bm \lambda}_{m}$ and $\omega_{m}$. These bandwidths are typically controlled by the bandwidth of the input signal and the size of the object being analyzed. 
\end{definition}  

\begin{theorem}
  Given the definitions above, $\exists$ constants $K_1$, $K_2$ and $M(max(\omega_{m},{\bm \lambda}_m))$, such that $\forall N \ge M$, \begin{equation}
 \varepsilon_N(\omega, {|\bm \lambda|})  \le \frac{z_t}{z_s}\left(z_t z_s\right)^N \left[\frac{K_1}{4N^2} \frac{z_t z_s}{\left(1-z_t^2 z_s^2\right)^2} + \frac{K_2}{4N^2} \frac{1}{1-z_t^2z_n^2}\right]
    \label{eq:bnd}
  \end{equation}
  where $z_t = \frac{\omega k_2}{2k_1 N}$ and $z_s = \frac{|\bll| k_2}{2k_1 N}$.
  \label{thm:thm1}
\end{theorem}
\begin{proof}
  The proof is by construction. \begin{subequations}
    The Fourier transform \begin{eqnarray}
      &&{\mathcal F}\left\{\sum_{l=0}^{N}\frac{(2l+1)}{2}k_1\frac{P_n\left(k_1\frac{|{\bf r}_0-{\bf r}'|}{c} + k_2\right)}{|{\bf r}_0-{\bf r}'|}P_n\left(k_1(t_0 -t')+k_2\right) \right\} \\&&=  \sum_{l=0}^N \frac{8\pi c}{\bll k_1} (2l+1) e^{-j\frac{\omega k_2}{k_1}} j_l\left(\frac{\omega}{k_1}\right)j_l\left(\frac{\bll c}{k_1}\right) \nonumber
      \label{eq:ft}
    \end{eqnarray}
    Thus, the error is the residual in the sum given by \begin{equation}
      \varepsilon_N(\omega, {\bm \lambda}) =  \sum_{l=N+1}^\infty (2l+1) \frac{8\pi c}{|\bll| k_1}  e^{-j\frac{\omega k_2}{k_1}} j_l\left(\frac{\omega}{k_1}\right)j_l\left(\frac{|\bll| c}{k_1}\right)
      \label{eq:newbnd}
    \end{equation}
    From the definition of the spherical Bessel function, we have 
    \begin{equation}
      j_l(z) = \sqrt{\frac{\pi}{2z}}J_{l+(1/2)}(z).
      \label{eq:sbesbnd}
    \end{equation}
Using the monotonicity relationship for cylindrical Bessel functions, for all $l+1 < z$. 
\begin{equation}
  J_{l+(1/2)}(z) \le \left(\frac{z}{l + \sqrt{l^2-z^2}}\right)^{(l)} e^{\left(\sqrt{l^2-z^2}\right)} ~\forall l<z-1,
  \label{eq:besbnd}
\end{equation}
    and combining with \eqref{eq:sbesbnd}, we have
    \begin{equation}
      j_l(z) \le \sqrt{\frac{\pi}{2}}\left(\frac{z}{l + \sqrt{l^2-z^2}}\right)^{(l)} e^{\left(\sqrt{l^2-z^2}\right)}.
      \label{eq:bndbdn}
    \end{equation}
    The quality of this bound is illustrated in Fig. \ref{fig:bnd}. Using \eqref{eq:bndbdn} and \eqref{eq:newbnd}, the bound on the error is 
\begin{eqnarray}
  \varepsilon_N(\omega, {|\bm \lambda|}) \le  \sum_{l=N+1}^\infty \left[(2l+1)\frac{C}{|\bll|} e^{-j\frac{\omega k_2}{k_1}}e^{\left(\sqrt{l^2-(\omega k_2/k_1)^2}  +\sqrt{l^2-(|\bll|c/k_1)^2} \right)}\nonumber \right.\\
  \left.\left(\frac{\omega}{l+\sqrt{l^2-\left(\frac{wk_2}{k_1}\right)^2}}\right)^l \left(\frac{|\bll|}{l+\sqrt{l^2-\left(\frac{|\bll| k_2}{k_1}\right)^2}}\right)^l \right].
      \label{eq:freqbnd}
    \end{eqnarray}
    Equation \eqref{eq:freqbnd} provides a tight bound on the error. For ease of representation, we can relax the bound by writing 
    \begin{eqnarray}
      \varepsilon_N(\omega, {|\bm \lambda|}) &\le&  \sum_{l=N+1}^\infty \left[(2l+1)\frac{C}{|\bll|} e^{-j\frac{\omega k_2}{k_1}}e^{\left(\sqrt{l^2-(\omega k_2/k_1)^2}  +\sqrt{l^2-(|\bll|c/k_1)^2} \right)}\nonumber \left(\frac{\omega}{2l}\right)^l \left(\frac{|\bll|}{2l}\right)^l \right] \nonumber \\
      &\le& \sum_{l=N+1}^\infty \left[(2l+1)\frac{C}{|\bll|} e^{-j\frac{\omega k_2}{k_1}}e^{2l}\nonumber \left(\frac{\omega}{2l}\right)^l \left(\frac{|\bll|}{2l}\right)^l \right]\nonumber \\
      &\le& \frac{C}{|\bll|} e^{-j\frac{\omega k_2}{k_1}}\sum_{l=N+1}^\infty \left[(2l+1) \left(\frac{e^2 \omega |\bll|}{4l^2}\right)^l\right]
      \label{eq:tmp}
    \end{eqnarray}
    The series in \eqref{eq:tmp} converges very rapidly and in particular, we have \begin{equation}
      \varepsilon_N(\omega, {|\bm \lambda|}) \le \frac{C}{|\bll|} e^{-j\frac{\omega k_2}{k_1}}\sum_{l=N+1}^\infty \left[(2l+1) \left(\frac{e^2 \omega |\bll|}{4N^2}\right)^l\right] 
      \label{eq:Nbnd}
    \end{equation}
    which is a standard arithmetico-geometric series with the sum\begin{equation}
      \frac{C\omega}{|\bll|} e^{-j\frac{\omega k_2}{k_1}}\left(\frac{e^2 \omega |\bll|}{4N^2}\right)^{N} \left[\frac{e^2 \omega |\bll|}{\left(4N^2-e^2 \omega |\bll|\right)^2} + \frac{2N+3}{\left(4N^2-e^2 \omega |\bll|\right)} \right]
      \label{eq:finbnd}
    \end{equation}
    Further setting $z_t = \frac{\omega k_2}{2k_1 N}$ and $z_s = \frac{|\bll| k_2}{2k_1 N}$, we have \begin{equation}
      \varepsilon_N(\omega, {|\bm \lambda|})  \le \frac{z_t}{z_s}\left(z_t z_s\right)^N \left[\frac{K_1}{4N^2} \frac{z_t z_s}{\left(1-z_t^2 z_s^2\right)^2} + \frac{K_2}{4N^2} \frac{1}{1-z_t^2z_n^2}\right]
      \label{eq:finalbnd}
    \end{equation}
  \end{subequations}
\end{proof}
\begin{figure}
  \centering
  \includegraphics[width=0.45\textwidth,height=0.38\textwidth]{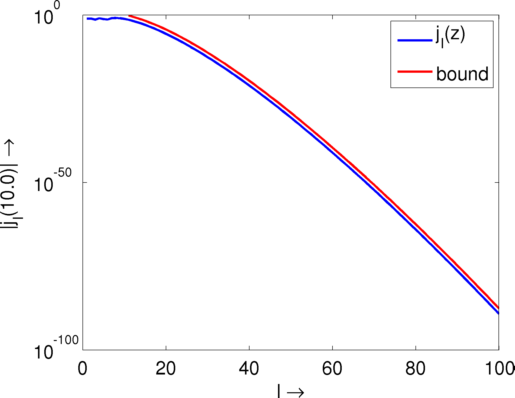}
  \caption{\label{fig:bnd} Bound on spherical Bessel functions. The figure shows $j_l(z)$ and the bound in \eqref{eq:bndbdn} at $z=10.0$ and as a function of $l$.}
\end{figure}

The bound on the error derived above implies that for $z_t<1$ and $z_s < 1$, the error decreases rapidly with $z_t$ and $z_s$. In turn, this also implies that for all values $ \omega \ge \omega_m$ and $|\bll| \ge |\bll|_m$, the error rapidly decreases with $N$, as long as $N > 2*max(\frac{\omega k_2}{k_1},\frac{|\bll k_2}{k_1})$. Indeed, in practice, it is sufficient to pick $N>2*min(\frac{\omega k_2}{k_1}, \frac{|\bll| k_2}{k_1})$. We utilize this result, to construct an efficient scheme to integrate the bilinear form derived in Section \ref{sec:stabtd}. In the next section, we will demonstrate numerical results obtained using this technique.  

\section{Numerical Results\label{sec:results}}
First, we demonstrate the construction of the spatio temporal bandwidths, and corresponding truncation limit $N$.  For simplicity and ease of demonstration, we will derive the truncation limits in temporal frequency; the spatial frequency bounds are identical. 

Consider a band limited signal that is convolved with a Lagrange polynomial of order $1$. This convolution is then approximated using expression in equation \eqref{eq:expn}, truncated to different values $N$. Figure \ref{fig:limits} shows the spectrum of the true convolution overlaid on approximations for different $N$. The domain of the expansion $\tball$ is chosen to be 3$\Delta t$ in width, and the temporal basis functions are chosen to be supported in $\omega_t= 2 \Delta_t$. Figure \ref{fig:err} shows the convergence in the norm of the error within the band as a function of $N$. As is clear from both the figures, the expansion can be efficiently truncated with minimal and controllable loss in accuracy.     
\begin{figure}
  \centering
  \subfloat[(Temporal) Spectrum of approximations to convolution\label{fig:limits}]{
  \includegraphics[width=0.45\textwidth,height=0.38\textwidth]{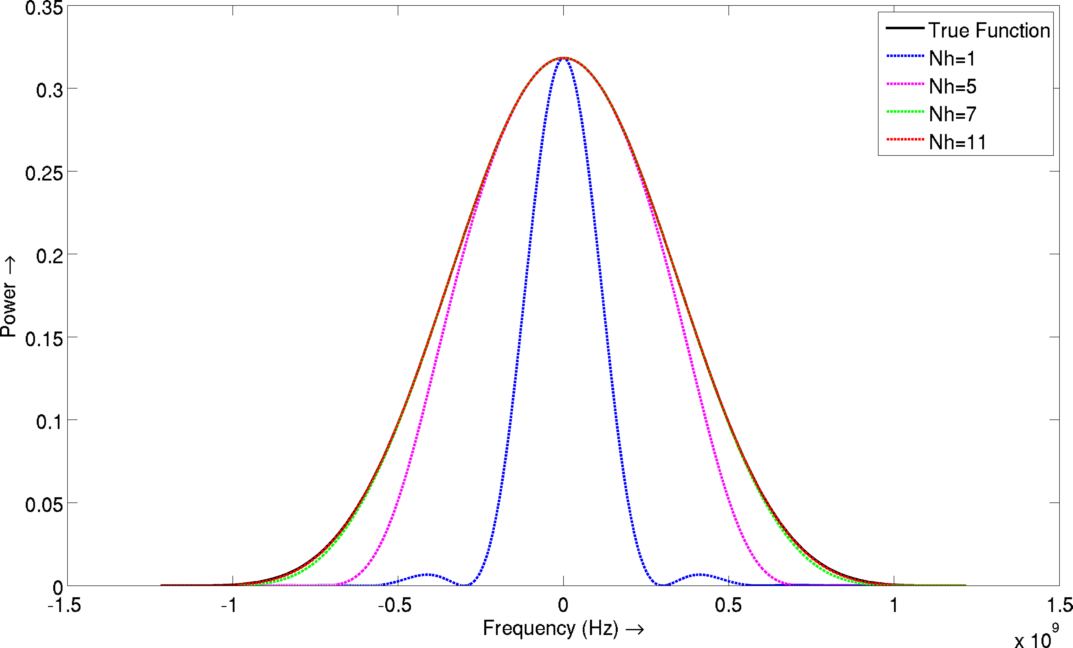}
  }
  \subfloat[Convergence of error with Nh\label{fig:err}]{
  \includegraphics[width=0.45\textwidth,keepaspectratio=true]{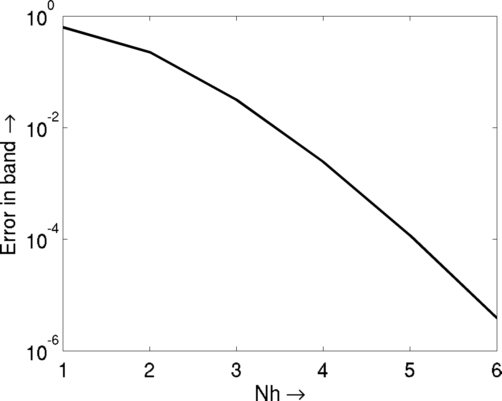}
  }
  \caption{Numerical realization of truncation bounds}
\end{figure}

Next, we present results that validate the full solution scheme developed in \cite{Pray2012} using the expansions derived in this paper. To this end, we consider scattering from a perfectly conducting sphere of radius $1.0m$, due to an incident plane wave. The incident field is a plane wave of the form 
\begin{equation}
	{\bf E}^{i}({\bf r},t)={\hat u}\cos(2\pi f_0t)e^{-(t-{\bf r}\cdot{\hat k} /c-t_p)^2/2\sigma^2 }~,
	\label{eq:einc}
\end{equation}
where $\hat{u}$ denotes the polarization vector, $f_0$ the center frequency, and 
${\hat k}$ the direction of propagation.  
The values  $\sigma$ and $t_p$ are calculated as $\sigma=3/(2\pi B)$ and $t_p=6\sigma$, where $B$ denotes the
bandwidth of ${\bf E}^i$ in Hz.  The incident power is calculated to be approximately 160 dB below the peak
at $f_{max}=f_0+B$ and $f_{min}=f_0-B$.  The sphere is discretized using curvilinear triangular elements (as described in \cite{Graglia1997}) of order $g=3$. Temporal basis functions are first order Lagrange polynomials and the expansion of the spatio-temporal convolution is truncated using the criteria derived in this paper, resulting in a maximum $N=7$. The incident field is wide band with a center frequency of $150MHz$ and bandwidth of $149.9MHz$. It is seen that the current on the surface of the sphere is stable for $20,000$ time steps (Figure \ref{fig:sphcur}) with $\alpha= 0.5$. Similarly, the far field observed due to this current, normalized to the incident field, can be used as a metric for comparison against an analytical frequency domain result. The excellent agreement with the analytical result at three different frequencies across the bandwidth validates the technique (Figure \ref{fig:rcs}). 
\begin{figure}[h!]
  \centering
  \subfloat[Current on Surface of sphere\label{fig:sphcur}]{
  \includegraphics[width=0.45\textwidth,keepaspectratio=true]{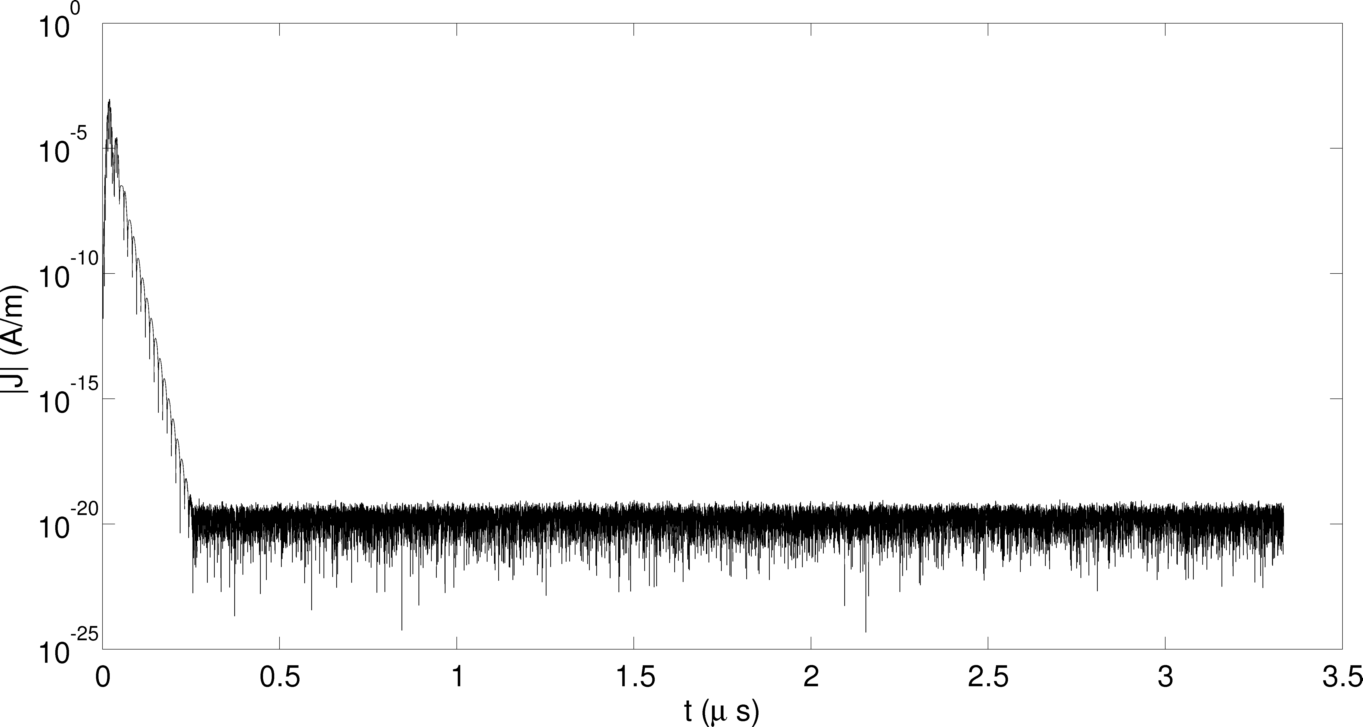}
  }
  \subfloat[Analytical validation of far field data\label{fig:rcs}]{
  \includegraphics[width=0.45\textwidth,keepaspectratio=true]{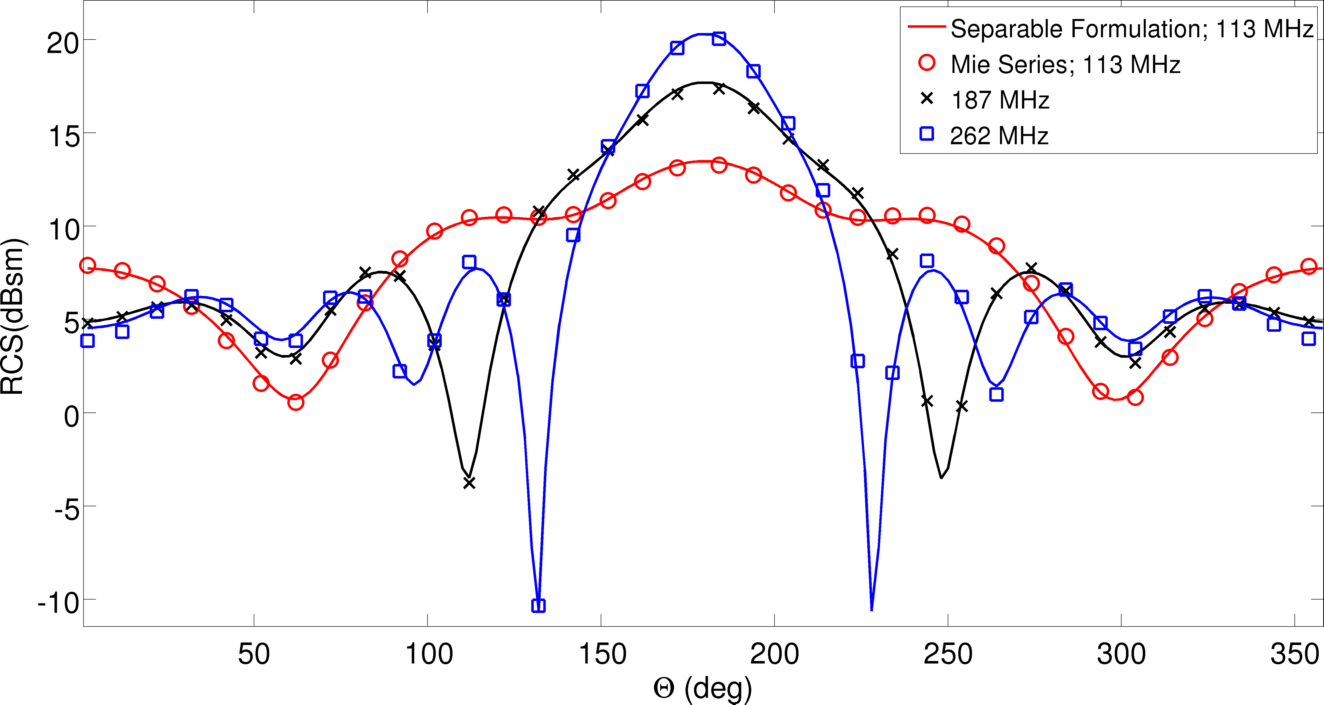}
  }
  \caption{Validation results for scattering from a sphere}
\end{figure}
\begin{figure}[h!]
  \centering
  \subfloat[Current on nacelle surface, for early time, comparison against standard MOT scheme\label{fig:nacellev5}]{
  \includegraphics[width=0.45\textwidth,keepaspectratio=true]{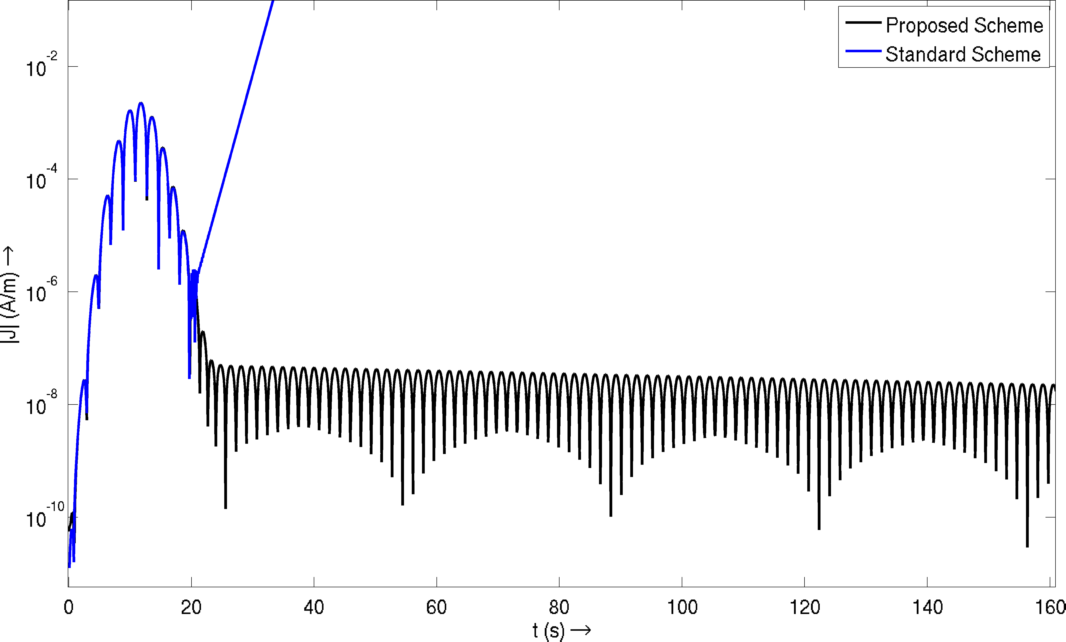}
  }
  \subfloat[Current on nacelle surface for long run time showing late time stability of proposed scheme\label{fig:nacelle}]{
  \includegraphics[width=0.45\textwidth,keepaspectratio=true]{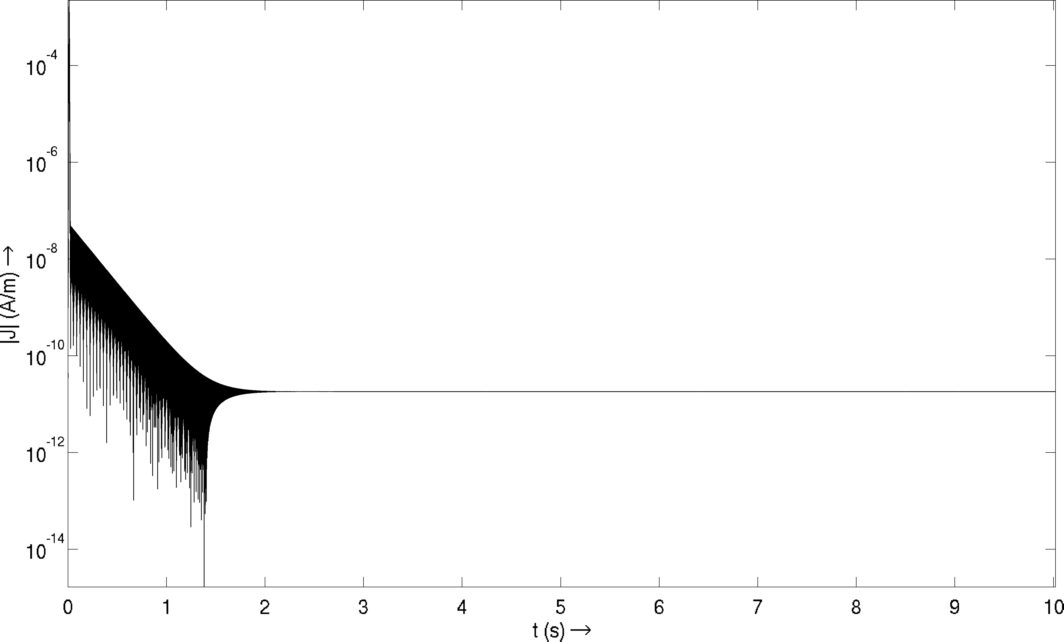}
  }
  \caption{Construction of stable solutions on a nacelle geometry}
\end{figure}

Finally we present a result that has proven challenging to stabilize using traditional MOT based schemes. We compute scattering from a nacelle geometry using the proposed scheme. The nacelle is constructed as a cylindrical shell with one end closed. The outer radius of the cylinder is $160cm$ and the inner radius is $100cm$. The length of the cylinder is $300cm$. To be realistic, the edges of the cylinder are filleted at a $20^o$ angle. The incident field is chosen to have the same functional dependence as in \eqref{eq:einc} with $f_0 = 1GHz$ and $B = 0.99GHz$. The incident field is polarized along $\hat{\theta}$ and directed along $\theta = 90^o$ and $\phi = 5^o$.  The maximum $N$ chosen across the geometry is $7$ and the minimum is $2$. The simulation is run for $100,000$ time steps, corresponding to $753$ transits.  While the stability of the scheme is proven for $\alpha > 0$, we have chosen $\alpha = 0$ in this case to illustrate the robustness of the approximation in terms of providing stable solutions for the more challenging case of the electric field integral equation. The early time data in Figure \ref{fig:nacellev5} shows the current on the surface of a nacelle geometry computed using the proposed scheme compared against the scheme developed in \cite{Shanker2009}. The figure not only validates the method but clearly indicates the instability of the standard MOT schemes.  Late time data in Figure \ref{fig:nacelle} demonstrates the capability of the scheme developed in the paper to generate stable solutions for extremely long simulation times ($753$ transits) for problems of scattering from challenging geometries.   
\section{Conclusions\label{sec:concl}}
In this work, we have presented a technique to numerically construct a stable solutions to time domain integral equations for electromagnetic scattering from perfectly conducting structures. We have derived a bound on the truncation error and provided a mechanism for numerically estimating the error. The technique has been numerically validated on a canonical structure using higher order tessellations \cite{Pray2012a,Pray2012b} and shown to be applicable to generate stable solutions using the EFIE for more challenging structures that defy stabilization via traditional methods 

\end{document}